\documentclass[aps,prx,reprint]{revtex4-1}
\usepackage{blindtext}
\usepackage{amssymb}
\usepackage{amsmath}
\usepackage{amsthm}
\usepackage{graphicx}
\newtheorem{lemma}{Lemma}

\newtheorem{proposition}{Proposition}

\newtheorem{prop}{Proposition}\def\PRO{\begin{prop}}\def\ORP{\end{prop}}
\newtheorem{coro}{Corollary}\def\COR{\begin{coro}}\def\ROC{\end{coro}}
\newtheorem{theo}{Theorem}\def\TH{\begin{theo}}\def\HT{\end{theo}}
\def\TH{\begin{theo}}\def\HT{\end{theo}}
\newtheorem{defi}[prop]{Definition}\def\DE{\begin{defi}}\def\ED{\end{defi}}

\newtheorem{lemme}[prop]{Lemma}\def\LE{\begin{lemme}}\def\EL{\end{lemme}}

\usepackage{color}

\newcommand{\beq}{\begin{equation}}
\newcommand{\eeq}{\end{equation}}

\begin{document}
\title{Secure quantum signatures using insecure quantum channels}
\author{Ryan Amiri$^1$}
\email{ra2@hw.ac.uk}
\author{Petros Wallden$^2$}
\author{Adrian Kent$^{3,4}$}
\author{Erika Andersson$^1$}
\affiliation{$^1$SUPA, Institute of Photonics and Quantum Sciences, Heriot-Watt University, Edinburgh EH14 4AS, United Kingdom\\
$^2$LFCS, School of Informatics, University of Edinburgh, 10 Crichton Street, Edinburgh EH8 9AB, United Kingdom\\
$^3$Centre for Quantum Information and Foundations, DAMTP, Centre for Mathematical Sciences,  University of Cambridge, Wilberforce Road, Cambridge, CB3 0WA, United Kingdom\\
$^4$Perimeter Institute for Theoretical Physics, 31 Caroline Street North, Waterloo, ON N2L 2Y5, Canada}

\begin{abstract}
Digital signatures are widely used in modern communication to guarantee authenticity and transferability
of messages. The security of currently used classical schemes relies on computational
assumptions. We present a quantum signature scheme that does not
require trusted quantum channels.  We prove {that} it is unconditionally
secure against the most general coherent attacks, {and show that it}   
requires the transmission of significantly fewer quantum states 
than previous schemes. We also show that the quantum channel noise threshold for
our scheme is less strict than for distilling a secure key using quantum key distribution. This
shows that ``direct'' quantum signature schemes can be preferable to signature schemes relying on
secret shared keys generated using quantum key distribution.
\end{abstract}

\maketitle

\section{Introduction}

Signature schemes allow for the exchange of messages from one sender
to multiple recipients, with the guarantee that messages cannot be
forged or tampered with. Additionally, messages can be transferred,
and cannot be repudiated. Transferability means that with a
probability that can be made arbitrarily close to one, if a message is
accepted by an honest recipient, 
it will also be accepted by another recipient if forwarded. 
The related requirement of non-repudiation means that, except with
probability that can be made arbitrarily small, a sender cannot later
successfully deny having sent a signed message.
Digital signatures are widely used for example in e-mail and
electronic commerce, and are considered to be one of the most
important inventions of modern cryptography. Unfortunately, the
security of commonly used signature protocols relies on the assumed
computational difficulty of certain problems. In the United States,
for example, there are currently three approved algorithms for
generating digital signatures -- RSA, DSA and ECDSA -- all of which
rely on the difficulty of finding discrete logarithms or factoring
large primes. With the advent of quantum computers, such assumptions
would no longer be valid. Given the importance of digital 
signatures, there is thus a strong motivation to develop practical
signature schemes whose security is unconditional, i.e. guaranteed by
the laws of physics, without any computational assumptions.

Unconditionally secure ``classical" signature schemes are possible, but need, at the very least, shared secret keys, and often also require a third party trusted by everybody (who effectively can provide each participant with secret information)~\cite{ChaumUSS, HanaokaSig, SwansonStinson, QDSwQKD}. Shared secret keys can of course be generated by quantum key distribution (QKD), so that an unconditionally secure signature scheme can proceed by first generating secret keys via QKD, and then running e.g. the protocol P2 in \cite{QDSwQKD}. Unconditionally secure ``direct" quantum signature schemes proceed without first distilling highly secure shared secret keys~\cite{Gottesman, QDSwMulti,QDSwMultiExp,QDSwQKD}. It is an open question what the best unconditionally secure signature schemes are, with respect to the number of quantum transmissions required per signed bit, trust assumptions, requirements on communication channels, and so on. In this paper, we explicitly demonstrate that ``direct" quantum signature schemes can have advantages over schemes relying on secret shared keys generated via QKD, by showing that the ``direct" scheme we propose can tolerate more noise in the quantum channels.

Previous quantum signature schemes \cite{QDSwMulti,QDSwMultiExp,QDSwQKD} improved on the original Gottesman-Chuang scheme \cite{Gottesman} by removing the need for quantum memory. 
In these quantum signature schemes, Alice encoded her signatures into quantum states and
sent a copy to both Bob and Charlie, who were only able to gain partial information on the overall signature due to the quantum nature of the states. However, the security analysis assumed authenticated quantum channels that did not allow eavesdropping. This strong and generally unrealistic assumption meant that a potential forger (Bob) only had access to his own copy of the signature states sent from Alice. In reality an adversarial Bob would be able to gain extra information on Alice's signature through eavesdropping on the signature states sent from Alice to Charlie.

Here we present a new quantum signature protocol, with three
improvements over previous protocols.  First, we remove all trust
assumptions on the quantum channels.  This is crucial for actual
practical use of quantum signature schemes. Second,
instead of Alice sending the same signature
states to Bob and Charlie, Bob and Charlie send {\it different} states to Alice, which leads to increased efficiency. 
This departs from the ``public-key" principle in the earlier quantum signature schemes.
Third, as already mentioned above, we show that in our direct quantum signature protocol, the noise threshold for the
Alice-Bob and Alice-Charlie quantum channels is less strict than for distilling a secret key using quantum key distribution
(QKD).

\section{The protocol}

We outline our protocol for three parties, with a sender, Alice, and
two receivers Bob and Charlie. Generalisation to more parties is
possible, but special care should be taken to address colluding
adversaries (see e.g. \cite{AWA15}). In the three-party scenario,
  at most one party can be dishonest, since two colluding dishonest
  parties can trivially cheat on the third party. In the multiparty
  scenario, the maximum number of dishonest parties will depend on the
  method of dispute resolution. If a majority vote is used to resolve
  disputes, then a majority of the participants must be honest. Also, 
  transferability and non-repudiation become identical in a three-party
   scenario when majority vote dispute resolution is used.
We assume that between Alice and Bob, and between Alice and Charlie
there exists authenticated classical channels as well as untrusted,
imperfect quantum channels. In addition, Bob and Charlie share a QKD
link which can be used to transmit classical messages in full
secrecy. The protocol makes use of a key-generating protocol (KGP)
performed in pairs separately by Alice-Bob and Alice-Charlie. The KGP
uses the noisy untrusted quantum channels, and generates two
correlated bit strings, one for the sender and one for the
receiver. When the noise level is below the prescribed
  threshold, the Hamming distance between the receiver's string and the
sender's string is smaller than the Hamming distance between any
string an eavesdropper could produce and the sender's string. The KGP
is further discussed below, after presenting the signature protocol
itself.  

The quantum signature protocol has two parts, a distribution stage, where the scheme is set up, and
a messaging stage, when messages are sent and signed. The distribution
stage involves both classical and quantum communication, but all communication in the messaging stage is classical.
We show how to sign a one-bit message. Longer messages can be signed for example by suitably iterating the one-bit protocol, as in \cite{WCRZ14}.

\subsection{Distribution stage}
\noindent(1) For each possible future message $m$=0 or 1, Alice uses
the KGP to generate four different length $L$ keys, $A^B_0, A^B_1,
A^C_0, A^C_1$, where the superscript denotes the participant with whom
she performed the KGP and the subscript denotes the future message, to
be decided later by Alice.
Bob holds the length $L$ strings $K^B_0, K^B_1$ and Charlie holds the 
length $L$ strings $K^C_0, K^C_1$. Due to the KGP, we know that
$A^B_0$ contains fewer 
mismatches with $K^B_0$ than does any string produced by an
eavesdropper, and the 
same applies to the other pairs of strings. Alice's signature for the future message $m$ will be $Sig_m = (A^B_m, A^C_m)$. 
Essentially, what will protect against forging is that only Alice knows a valid signature for a message $m$. \\

\noindent (2) For each future message, Bob and Charlie symmetrise
their keys by choosing half of the bit values in their $K^B_m, K^C_m$
and sending them (as well as the corresponding positions) to the other
participant using the Bob-Charlie secret classical channel. As
explained below, this ensures that Alice cannot make Bob and Charlie
disagree on the validity of a signature if a message is forwarded from Bob to Charlie or vice versa in the messaging stage.  
If Bob (or Charlie) chooses to forward an element of $K^B_m$ (or $K^C_m$) in the distribution stage, he will not further use it to check the validity of a
signature. They will only use the bits they did not forward and
those received from the other participant~\cite{footnote1}.  We denote
their symmetrised keys by $S^B_m$ and $S^C_m$, with the superscript
indicating whether the key is held by Bob or Charlie. Bob (and
Charlie) will keep a record of whether an element in $S^B_m$ (
  $S^C_m$) came
directly from Alice or whether it was forwarded to him by Charlie (or
Bob).

At this point in the protocol, Bob's and Charlie's symmetrised strings each
contain half of $K^B_m$ and half of $K^C_m$. For
each future possible message $m$, Bob and Charlie each have a bit
string of length $L$, and Alice has no information on whether it
is Bob's $S^B_m$ or Charlie's $S^C_m$ that contains a particular
element of the string $(K^B_m, K^C_m)$, {which has length $2L$}. This protects against repudiation. Bob has access to all of $K^B_m$ and half of $K^C_m$, but, even
if he is dishonest, he does not know the half of $K^C_m$ that Charlie
chose to keep. This protects against forging by Bob (and similarly for forging by Charlie).

\subsection{Messaging stage}

\noindent (1) To send a signed one-bit message $m$, Alice sends $(m, Sig_m)$ to the desired recipient (say Bob).\\

\noindent (2) Bob checks whether $(m, Sig_m)$ matches his $S^B_m$ and records the number of mismatches he finds. He separately checks the part of his key received directly from Alice and the part of the key received from Charlie. If there are fewer than $s_a(L/2)$ mismatches in both halves of the key, where $s_a<1/2$ is a small threshold determined by the parameters and the desired security level of the protocol, then Bob accepts the message.\\ 

\noindent (3) To forward the message to Charlie, Bob forwards the pair $(m, Sig_m)$ that he received from Alice.\\

\noindent (4) Charlie tests for mismatches in the same way, but in order to protect against repudiation by Alice he uses a different threshold. Charlie accepts the forwarded message if the number of mismatches in both halves of his key is below $s_v(L/2)$ where $s_v$ is
another threshold, with $0 < s_a < s_v < 1/2$. That the recipients have to use different thresholds or acceptance criteria for messages received directly from the sender and for forwarded messages is a general and necessary feature of unconditionally secure signature schemes~\cite{SwansonStinson, AWA15}. More generally in a multiparty situation, thresholds depend on how many times a message has been forwarded, and the level of mismatches will determine how many times a message can subsequently be forwarded.\\

\section{Key generation protocol}

We now describe how two parties, for now called Alice and Bob, perform the KGP. Essentially, Alice and Bob perform the quantum part of QKD to generate raw keys, but do not proceed to error correction or privacy amplification. This means that Alice and Bob will generate different (but correlated) strings that are not entirely secret. These keys will be the $A^B_i$, $K^B_i$ described above. 
Even though the KGP builds on QKD, the security analysis for the KGP does not follow directly from the security of the QKD protocol. This is because the goal of an adversary in the signature protocol is different from that of an eavesdropper in QKD.
For the signature protocol, what matters is the number of mismatches with a recipient's key; for QKD, what matters is the information an eavesdropper can hold about the key. From the bound on an eavesdropper's min-entropy in QKD, we show how to bound the number of mismatches a forger in our signature protocol can achieve.
Our aim is to show that $d(A^B_i, K^B_i)<d(E_{\text{guess}}, K^B_i)$ except with negligible probability, where $d(.,.)$ is the Hamming distance and $E_{\text{guess}}$ is Eve's attempt at guessing $K^B_i$ (and it may be that Eve is Charlie). In addition to proving the security of the KGP {itself}, the security of the signature protocol {(in which the KGP is used as a subprotocol)} will be proven below in Sec. \ref{sec:Security}. 

In what follows, the underlying QKD protocol upon which the KGP is built will be the prepare-and-measure decoy-state BB84 protocol using weak coherent pulses, described in \cite{Lim}.
Apart from the post-processing, another difference is that here it is Bob who prepares the states and sends them along the quantum channel to Alice. This may not be necessary, but simplifies the security analysis in that a dishonest Alice cannot send the recipients Bob and Charlie entangled states.
Specifically, when the KGP is performed by Bob and Alice, we assume that Bob has a phase-randomised source of
coherent states. The intensity of each light pulse is chosen by Bob to
be either $u_1$, $u_2$, or $u_3$, where $u_1>u_2>u_3$. The
intensities are chosen with probabilities $(p_{u_1}, p_{u_2},
p_{u_3})$. As in \cite{Lim}, we use all intensity levels for key generation. To encode information, Bob randomly selects one of four
possible polarisation states -- $|0_Z\rangle, |1_Z\rangle$ ($Z$ basis)
and $|0_X\rangle=1/\sqrt{2}(|0_Z\rangle+|1_Z\rangle),
|1_X\rangle=1/\sqrt{2}(|0_Z\rangle-|1_Z\rangle)$ ($X$
basis). The $X$ and $Z$ bases are chosen with probabilities $p_X\geq
1/2$ and $p_Z =1-p_Z\leq 1/2$ respectively. {The asymmetric probabilities for the two bases increases the efficiency of the protocol~\cite{LoChauArdehali}. Intensities and
states are chosen independently by Bob to avoid correlations between
intensity and information encoding.} Alice also independently chooses the $X$ and $Z$
 measurement bases with probabilities $p_X$ and $p_Z$ respectively. 
 
For each state sent by Bob, Alice obtains one of four possible outcomes $\{0,1,\emptyset, d\}$ where $0$ and $1$ are the bit values, $\emptyset$ represents no detection and $d$ is a double click event. In the case of double clicks, Alice randomly chooses a bit value.  Alice and Bob {then} announce their basis and intensity choices over an authenticated classical channel. If states are transmitted and then
measured in different bases, or if there is no detection, they are discarded (sifting). The protocol is continued until a sufficient number of measurement outcomes have been obtained for each basis and intensity choice. A raw key is generated by choosing a random sample of size $L+k$ of the $X$ basis counts. The bit string generated by Bob is split into four parts $(V_B, Z_{B}, X_{B, \text{keep}}, X_{B, \text{forward}})$. Alice will hold corresponding strings but with the subscript $B$ replaced by $A$. The $V$ strings have length $k$ and are generated from $X$ basis measurements. They are used to estimate the correlation between Alice{'s} and Bob's strings generated from X basis measurements, after which they are discarded. The $Z$ strings are generated from $Z$ basis measurements. They will be used to quantify the level of eavesdropping by Eve. Roughly speaking, due to the complementary nature of the $X$ and $Z$ bases, eavesdropping must affect the correlations Alice and Bob would expect to see in their states and measurement results, and they can use a measure of their correlations to find a quantitive bound on the min-entropy the eavesdropper has on Bob's $X$ strings. The {two} $X_B$ strings have length $L/2$ and together make up Bob's key, $K^B_i$. {Bob will forward $X_{B,\text{forward}}$ to Charlie (who could in fact be Eve) and will keep the other string, $X_{B,\text{keep}}$, for himself. Bob will no longer use the bits in $X_{B,\text{forward}}$.}

It should be stressed that in signature schemes it
cannot be assumed that Alice and Bob are honest. This is another difference from standard QKD. However, as
explained below, neither {of them} gain from dishonesty during the KGP, and
therefore we can assume that they behave honestly during the KGP stage.

In what follows we will consider the finite case, that is, a finite number of states are sent and measured, with Eve
allowed to perform the most general attack permissible by quantum mechanics {-- so-}called ``coherent" attacks. This means {that} Eve can perform entangling operations on any/all 
states sent over the quantum channel, and at any later time make a general measurement on an ancilla system kept in quantum memory.

Our strategy will be to find Eve's information in terms of her smooth min-entropy, and use that to bound the probability that she can make a signature declaration containing fewer than a certain number of errors. To begin with then, we must find Eve's smooth min-entropy on Bob's key $X_{B, \text{keep}}$. We follow \cite{Lim} and find Eve's smooth min-entropy in the same manner as for decoy state QKD, with the only difference being that here, Bob gives the extra information $X_{B, \text{forward}}$ to Eve. However, since Bob does not subsequently use this part of the key, this can be treated in the same manner as is done for the $V$ string sacrificed for parameter estimation, as detailed more explicitly in Proposition 6 of \cite{TomamichelQcrypt}. Essentially, Eve's smooth min-entropy on $X_{B, \text{keep}}$ can be found using entropic uncertainty relations based on the level of correlation between $Z_B$ and $Z_A$. For ease of notation, we will simply write $X$ instead of $X_{B, \text{keep}}$, and we will denote its length by $n$. 
Eve also gains information from the classical communication between Alice and Bob, which is assumed to be public but authenticated. The classical random variables $V$, $\Theta^n$ and $X_{B, \text{forward}}$ represent the information gained by Eve from parameter estimation, basis declarations in the sifting step and, if Eve is Charlie, the forwarding of $X_{B, \text{forward}}$ by Bob, respectively.

We gather all of Eve's information into one {quantum} system {living in the Hilbert space $\mathcal{H}_E$.
This comprises the space containing Eve's ancilla quantum system following her coherent attack, $\mathcal{H}_{E^\prime}$, as well as the spaces containing the states encoding the strings $V$, $\Theta^n$, and $X_{B, \text{forward}}$, which we assume are known to Eve.}
 As in Appendix B of \cite{Lim}, we find
\beq \label{eq:coh_min_entr}
\begin{split}
& H^\epsilon_{\text{min}}(X|E) \gtrapprox s^L_{X,0} + s^L_{X,1}\big[ 1 - h(\phi^U_{X,1}) \big],
\end{split}
\eeq
Where the inequality holds up to a small additive term proportional to $\log(1/\epsilon)$. Here $s^L_{X,0}$ and $s^L_{X,1}$ are the number of pulses reaching Alice which come from {$0$- and $1$-}photon pulses respectively, and which make up the entries in the string $X$. $\phi^U_{X,1}$ is the phase error rate in $X$ basis measurements coming from single{-}photon pulses. The {superscripts $U$ and $L$} represent worst{-}case scenario estimates consistent with parameter estimation performed on a finite sample (see Appendix \ref{App:estimates}).

Now the question is, given Eve's smooth min-entropy, is it possible to bound the number of errors she is likely to make when guessing Bob's key?
\begin{proposition}
Suppose that Bob and Eve share the state $\rho_{XE}$ where, as above, $X$ is the $n$-bit string representing the part of Bob's key that is not forwarded to Charlie/Eve, and $E$ is the correlated quantum system held by Eve, including all information gained from classical communications. Then, for any eavesdropping strategy, Eve's average probability of making at most $r$ mistakes when guessing $X$ can be bounded as
\begin{equation} \label{eq:result2}
\langle p_r \rangle \leq \sum^r_{k=0} \binom{n}{k} 2^{-H^\epsilon_{min}(X|E)_{\rho}} + \epsilon.
\end{equation}
\end{proposition}
The proof of this proposition is given in Appendix \ref{App:lemma}. We can further use Markov's inequality to say that for any $a>0$,
\begin{equation} \label{eq:result3}
\mathbb{P}(\text{Eve makes fewer than $r$ mistakes}) := p_r \leq a 
\end{equation}
except with probability at most
\beq
\epsilon_F := \frac{1}{a} \left( b^r_n 2^{-H^\epsilon_{min}(X|E)_{\rho}} + \epsilon\right)
\eeq
Where $b^r_n := \sum^r_{k=0} \binom{n}{k}$ and for large $n$, we have $b^r_n \approx 2^{nh(r/n)}$. So we have found a bound on the probability of Eve making fewer than $r$ mistakes in terms of her smooth min-entropy. Using this, as well as \eqref{eq:coh_min_entr} for the min-entropy, we find
\beq \label{eq:forgeprob}
\begin{split}
\epsilon_F = \frac{1}{a} \left[2^{-n\left\{c^L_{X,0} + c^L_{X,1} [1-h(\phi^U_{X,1})] - h(r/n)\right\}} + {\epsilon}\right],
\end{split}
\eeq
where $c^L_{X,i} := s^L_{X,i}/n$ is the lower bound on the count rate for X basis pulses containing $i$ photons. The equation above should technically have an approximation sign rather than an equality since we have used the approximate bound on the min-entropy from Eq. \eqref{eq:coh_min_entr}. It can be made exact by including the terms proportional to $\log(1/\epsilon)$ in the min-entropy, however, for simplicity we have neglected such terms in the main body of the paper. The condition
\beq
c^L_{X,0} + c^L_{X,1} [1-h(\phi^U_{X,1})] - h(r/n) > 0
\eeq
determines whether or not Eve is able to make fewer than $r$ errors with non-negligible probability. If the condition holds, $n$ can be increased to make Eve's probability of making fewer than $r$ errors arbitrarily small. We define $p_E$ by the equation
\beq \label{eq:EveError}
c^L_{X,0} + c^L_{X,1} [1-h(\phi^U_{X,1})] - h(p_E) = 0.
\eeq
The meaning of this is that $p_E$ is the minimum rate at which Eve can make errors (except with negligible probability). Suppose the error rate on X basis measurements between Alice and Bob is upper bounded as $e^U_X$. As long as $p_E>e^U_X$, there exists a choice of parameters and a sufficiently large signature length which makes the protocol secure (see Section \ref{sec:Security}). Equivalently, QDS is possible as long as
\beq \label{eq:coh_range}
c^L_{X,0} + c^L_{X,1} [1-h(\phi^U_{X,1})] - h(e^U_X) > 0.
\eeq 

\section{Security analysis}
\label{sec:Security}

We will now prove the security of the main signature protocol, i.e. the robustness (probability of an honest run aborting), security against forging (probability that a recipient generates a signature, not originating from Alice, that is accepted as authentic) and repudiation
(or transferability) (probability that Alice generates a signature that is accepted by Bob but then when forwarded, is rejected by Charlie). In what follows, we assume that Alice-Bob and Alice-Charlie have each used the KGP to generate length $L$ bit strings to use in the QDS protocol described above.

\emph{(a) Robustness}. Bob rejects a signed message if the 
{$L/2$} bits received from either Alice or Charlie have a mismatch rate higher than $s_a$ with Alice's signature. From parameter estimation performed on the strings $V_A, V_B$ (whose length we denote by $k$), Alice and Bob obtain an estimate of the error rate they have {with respect to each other, for}
the strings they generated in the $X$ basis. We denote the observed error rate by $\tilde{e}_X$. Using the Serfling inequality \cite{Serfling}, we can bound the actual error rate between the strings $X_{A, \text{keep}}$ and $X_{B, \text{keep}}$ (which we denote as $e_X$) by
\beq \label{eq:pe}
e_X \leq \tilde{e}_X + \delta := e^U_X,
\eeq
where
\beq
\delta := \sqrt{\frac{\ln(1/\epsilon_{PE})}{2k}\left(1-\frac{k-1}{n}\right)}.
\eeq
This bound holds except with probability $\epsilon_{PE}$. It can be seen that for any fixed choice of $\delta$, the failure probability $\epsilon_{PE}$ decays exponentially fast in the parameter $k$.
Let $e^U_{X,B}$, $e^U_{X,C}$ be the worst-case error rates Alice has from performing separate KGP's with Bob and Charlie respectively. Set $e^U_{X} := \max\{e^U_{X,B}, e^U_{X,C}\}$ and choose $s_a$ such that $s_a > e^U_{X}$. The Serfling inequality tells us that the true error rate between Alice's and Bob's keys will be less than $e^U_X$ except with probability at most $\epsilon_{PE}$, so the probability of an honest abort is simply
\begin{equation} \label{eq:honab}
\mathbb{P}(\text{Honest Abort}) \leq 2\epsilon_{PE},
\end{equation}
where the factor of 2 occurs since the abort can be due to either the states received from Alice or the states received from Charlie.

\emph{(b) Security against forging.} It is easier for either Bob or Charlie to forge than for any other external party, and we will therefore consider forging by an internal party. 
In order to forge a message, Bob must give a declaration $(m, Sig_m)$ to Charlie that has fewer than $s_vL/2$ mismatches with the unknown (to Bob) half of $S^C_m$ sent directly from Alice to Charlie, and fewer than $s_vL/2$ mismatches with the half he himself forwarded to Charlie. An adversarial Bob will obviously be able to meet the threshold on the part he forwarded to Charlie. We therefore consider only the unknown half that was received directly from Alice. If parameter estimation is successful in the KGP, then we know the worst-case (maximum) rate at which Alice will make errors with Charlie's key; we denote it by $e^U_{X}$. From Eq. \eqref{eq:EveError}, we also know the minimum rate at which Bob will make errors with Charlie's key; we denote it by $p_E$.

Assuming \eqref{eq:coh_range} holds, we choose $s_v$ such that $e^U_{X} < s_v < p_E$. In this case, Charlie will likely accept a legitimate signature sent by Alice, since the upper bound on their error rate, $e^U_X$, is less than the threshold $s_v$. On the other hand, Charlie will likely reject any dishonest signature declaration by Bob, since the probability of Bob finding a signature with an error rate smaller than $s_v$ is restricted by equations \eqref{eq:result3} and \eqref{eq:forgeprob} as
\beq \label{eq:forge1}
\mathbb{P}(\text{Eve makes fewer than $s_vL/2$ errors}):=p_{s_vL/2} \leq a
\eeq
except with probability at most
\beq
\epsilon_F := \frac{1}{a} \left[2^{\frac{L}{2}h(s_v) - H^\epsilon_{\text{min}}(X|E)_\rho } + {\epsilon}\right] \\
\eeq
Let us suppose that if any of the parameter estimation procedures fail (so for example if $e^U_X$ is not a good upper bound), or if $p_{s_vL/2} \geq a$, then Bob is able to successfully forge with certainty. We are then able to bound Bob's probability of successfully forging as
\begin{equation} \label{eq:forge}
\mathbb{P}(\text{Forge}) \leq a+\epsilon_F + 8\epsilon_{PE}.
\end{equation}
This equation is valid for any choice of $a,\epsilon, \epsilon_{PE} > 0$ and so can be made arbitrarily small by increasing $L$. 
The addition of $8\epsilon_{PE}$ is to account for the possibility of the upper/lower bounds failing on any of the $e_X$, $s_{X,0}, s_{X,1}$ or $\phi_{X,1}$ (see Appendix \ref{App:estimates}).
Note that security against an adversarial Bob
derives entirely from the Alice-Charlie KGP, in which Bob is already
assumed to be an adversary. Thus, any dishonesty on Bob's part during
the Alice-Bob KGP cannot help him to forge. Exactly the same arguments apply when Charlie is the forger.

\emph{(c) Security against repudiation.} 
Alice aims to send a declaration $(m, Sig_m)$ which Bob will accept and, when forwarded, Charlie will reject. To do this, we must have that Bob accepts both the elements that Alice sent directly to him and the elements that Charlie forwarded to him. In order for Charlie to reject he needs only reject one of either the elements he received from Alice, or the elements Bob forwarded to him. Intuitively, security against repudiation follows because of the symmetrisation performed by Bob and Charlie using the secret classical channel. Even if Alice knows and can control the error rates between $A^B_m$, $A^C_m$ and $K^B_m$, $K^C_m$, she cannot control whether the errors end up with Bob or Charlie. After symmetrisation, the keys $S^B_m$ and $S^C_m$ will each have the same expected number of errors. To repudiate, 
one must contain significantly more errors than the other. Using results in \cite{Chvatal}, we can bound this probability as
\begin{equation} \label{eq:rep}
\mathbb{P}(\text{Repudiation}) \leq 2\exp\left[-\frac{1}{4}(s_v-s_a)^2L\right].
\end{equation}
For a more formal proof, please see Appendix \ref{App:AppendixB}. Note that security against repudiation derives entirely from the symmetrisation performed by Bob and Charlie, in which Alice plays no part. Even if Alice can control the choices of $s_a$, $s_v$ by manipulating the error rates achieved during the Alice-Bob KGP and the Alice-Charlie KGP, the choice of $L$ depends on $s_a$ and $s_v$ and the protocol will be secure for any valid choice.

\section{Comparison to QKD} 

For the BB84 protocol performed using decoy states as described above, with a finite number of states sent and received, Appendix B of \cite{Lim} gives the length of the extractable secret key as
\begin{equation} \label{eq:9}
\begin{split}
l &= \left\lfloor s^L_{X,0} + s^L_{X,1}\left[1-h(\phi^U_{X,1})\right] - \lambda_{EC} - \log \frac{2}{\epsilon_{cor}(\alpha_2\alpha_3\nu)^2} \right\rfloor \\
& \approx s^L_{X,0} + s^L_{X,1}\left[1-h(\phi^U_{X,1})\right] - \lambda_{EC},
\end{split}
\end{equation}
where $\epsilon_{cor}$ and $\nu$ are constants related to the possibility of failure of error correction and privacy amplification. The term $\lambda_{EC}$ represents the information leaked to Eve during error correction. It depends on the specific implementation, but must be greater or equal to $nh(e^U_X)$, where $n$ is the size of the bit string being corrected. In practice, error correction will not be perfect and it is common to write $\lambda_{EC} = nf_{EC}h(e^U_X)$ where $f_{EC}$ is a leakage parameter. To perform error correction, the total key is split into blocks and the leakage parameter, $f_{EC}$, depends on this block size, but not the overall length of the key. Increasing the block size reduces $f_{EC}$ at the cost of decreasing the efficiency of the error correction protocol. Estimates of $f_{EC}$ for practically feasible error correction is an area of active research \cite{IR}, though it is commonly estimated to be in the range $1.11 - 1.2$, regardless of the length of the total key being distilled. For example, \cite{Finite_renner} assumes $f_{EC}=1.2$ based on the performance of error-correcting codes in use at ID Quantique. Rewriting \eqref{eq:9}, we obtain
\begin{equation} \label{eq:Finite_qkd}
l \approx n \left\{c^L_{X,0} + c^L_{X,1}\left[1-h(\phi^U_{X,1})\right] - f_{EC}h(e^U_X)\right\}.
\end{equation}
Comparing equations \eqref{eq:coh_range} and
\eqref{eq:Finite_qkd}, we immediately see that there are Alice-Bob and
Alice-Charlie quantum channels for which quantum signatures are
possible and yet practical QKD gives a zero key generation rate. As stated above, $f_{EC}$ is
independent of $n$ and so cannot be decreased by simply increasing the
size of the total key. The important point is that because the quantum
signature scheme omits the inefficient process of error correction,
there should always be some region where quantum signature generation is possible but QKD is not.

\section{Discussion}

In this paper we have presented a quantum 
signature protocol and proven it unconditionally secure
against coherent attacks.   It improves on previous quantum signature
protocols by removing all trust assumptions on the quantum channels
between participants. One might have expected that a protocol that
does not rely on these assumptions must necessarily be less
efficient.  However, our protocol also significantly reduces the
length of the signature needed to sign a
message. To facilitate comparison to previous
quantum signature protocols, suppose one wants the probabilities in
\eqref{eq:honab}, \eqref{eq:forge}, \eqref{eq:rep} to all be below
$10^{-4}$.  Using realistic experimental quantities (taken from \cite{Lucamarini}), we estimate that a signature length of $L=7.71\times 10^5$ {(for each of the  possible one-bit messages 0 and 1)} is required to securely sign a one-bit message, sent over a distance of 50~km. This would
require Bob/Charlie to transmit approximately $6.3\times10^8$ quantum states (per bit to be signed)
to Alice during their KGP's (see Appendix \ref{App:AppendixC}). We compare this to previous quantum
signature protocols which required $O(10^{10})$ states to be
transmitted to achieve the same level of security over 1~km
\cite{QDS_Exp}.

The increase in efficiency is largely due to the fact that in our
protocol Bob and Charlie send Alice \textit{different} states,
whereas in previous protocols Alice sent Bob and Charlie the same signature states. In those
protocols, even without any eavesdropping, a potential forger has
access to a legitimate copy of each of the states Alice sent to the
participants, and thus to reach the same levels of security requires longer
signatures. Moreover, when generalising to $N$ participants with up to
$t$ dishonest parties, potentially colluding forgers are even more powerful, since
 they may have $t$ legitimate copies of each state. In
our protocol, where different states are sent by each participant, this
problem is evaded. The only source of information for a potential
forger is by eavesdropping on the quantum channels, an activity
{not even considered} 
in previous protocols due to the assumption of
``authenticated" quantum channels.

We also showed that the noise threshold in the quantum
channels connecting Alice-Bob and Alice-Charlie is in
practice less strict for
quantum signatures than for distilling a secret key using QKD. For some quantum channels, therefore, quantum signature
protocols that use QKD (e.g. P2 of \cite{QDSwQKD}) are not possible,
while our direct quantum protocol remains possible. This is an example
that direct quantum protocols are sometimes preferable to
protocols relying on secret shared keys generated using standard QKD,
and highlights that quantum signature protocols are not in general merely a
direct combination of QKD protocols and classical post-processing.

\begin{acknowledgments} The authors would like to thank Vedran Dunjko for discussions. This work was supported by the UK Engineering and Physical Sciences Research Council (EPSRC) under EP/G009821/1 and EP/K022717/1. P.W. gratefully acknowledges support from the COST Action MP1006. A.K. was partially supported by a grant from FQXi and by Perimeter Institute for Theoretical Physics. Research at Perimeter Institute is supported by the Government of Canada through Industry
Canada and by the Province of Ontario through the Ministry of Research and Innovation.
\end{acknowledgments}


\appendix

\section{Finite-size estimates}\label{App:estimates}

Equation \eqref{eq:coh_min_entr} contains three quantities to estimate -- $s^L_{X,0}$ and $s^L_{X,1}$, which are estimates of the number of counts (sent and measured in the X basis) containing zero and one photon respectively; and $\phi^U_{X,1}$ which is an estimate of the phase error rate in the X basis counts. 

As in \cite{Lim}, we have
\beq \label{eq:A4}
s_{X,0} \geq \frac{\tau_0}{u_2 - u_3}\left(\frac{u_2e^{u_3}n^*_{X,u_3}}{p_{u_3}} - \frac{u_3e^{u_2}n^*_{X,u_2}}{p_{u_2}}\right)
\eeq
where $n^*_{X,u_k}$ is the number of counts (from states sent and measured in the X basis) coming from pulses with intensity $u_k$, and $\tau_{n} := \sum_{u_k} p_{u_k}e^{-u_k}u^n_k/n!$. This formula is valid in the asymptotic limit where the number of counts will be equal to the expected value. In the finite setting, we cannot know with certainty the actual value of $n^*_{X,u_k}$. This is because once the raw key is generated, we randomly choose a finite sample of $L+k$ states from X basis counts. Nevertheless we are able to bound $n^*_{X,u_k}$ {from} above and below with high probability using the statistics observed in parameter estimation. Specifically, if $n_{X,u_k}$ are the observed statistics, Hoeffdings inequalities \cite{Hoeffding} give
\beq \label{eq:A7}
\begin{split}
&n^-_{X,u_k}:= n_{X,u_k}-\delta(L+k, \epsilon_{PE}) \leq n^*_{X,u_k} \\
&n^+_{X,u_k}:= n_{X,u_k}+\delta(L+k, \epsilon_{PE}) \geq n^*_{X,u_k}.
\end{split}
\eeq
These bounds hold with probability at least $1-\epsilon_{PE}$, where $\delta(n_X,\epsilon_{PE}) := \sqrt{n_X\ln(1/\epsilon_{PE})/2}$. Replacing the $n^*_{X,u_k}$ in Eq. \eqref{eq:A4} by the corresponding worst-case finite-size estimate leads to a finite-size lower bound on $s_{X,0}$, which we call $s^L_{X,0}$, and which holds with probability at least $1-2\epsilon_{PE}$.

Similarly, we can bound $s^L_{X,1}$ as
\beq \label{eq:A5}
\begin{split}
s^L_{X,1} \geq & \frac{u_1\tau_1}{u_1(u_2-u_3) - (u^2_2 - u^2_3)}\Bigg[\frac{e^{u_2}n^-_{X,u_2}}{p_{u_2}} - \frac{e^{u_3}n^+_{X,u_3}}{p_{u_3}}\\
& + \frac{u^2_2-u^2_3}{u^2_1}\left(\frac{s^L_{X,0}}{\tau_0}-\frac{e^{u_1}n^+_{X,u_1}}{p_{u_1}}\right)\Bigg].
\end{split}
\eeq

The $X$ basis phase errors are not directly observed in the protocol. Instead, we relate $\phi^U_{X,1}$ to the bit error rate in the $Z$ basis. As in Appendix B of \cite{Lim}, we have
\beq \label{eq:B2}
\phi^U_{X,1} \leq \frac{v^U_{Z,1}}{s^L_{Z,1}} + \gamma\left(\alpha_1, \frac{v^U_{Z,1}}{s^L_{Z,1}}, s^L_{Z,1}, s^L_{X,1}\right),
\eeq
where $v^U_{Z,1}$ is the upper bound on the number of bit errors in Z basis counts coming from single photon pulses, and
\beq
\gamma(a,b,c,d):= \sqrt{\frac{(c+d)(1-b)b}{cd\ln 2}\log\left[\frac{c+d}{cd(1-b)b}\frac{1}{a^2}\right]},
\eeq
where $\alpha_1$ comes from the calculation of the min-entropy given in \cite{Lim}, and is such that $\alpha_1>0, \epsilon>2\alpha_1+\alpha_2+\alpha_3$. Here $\epsilon$ is the smoothing parameter in the smooth min-entropy.

All quantities on the right hand side of Eq. \eqref{eq:B2} are known, except $v^U_{Z,1}$ which we can find as
\beq
v^U_{Z,1} \leq \frac{\tau_1}{u_2 - u_3} \left(\frac{e^{u_2}m^+_{Z,u_2}}{p_{u_2}} - \frac{e^{u_3}m^-_{Z,u_3}}{p_{u_3}}\right),
\eeq
where the $m^\pm_{Z,u_k}$ are the upper and lower bounds on the true number of bit errors from Z basis counts of intensity $u_k$. These are found from the observed number plus finite-size variations, similar to Eq. \eqref{eq:A7}.

\section{Proof of Proposition 1}\label{App:lemma}
In order to guess $X$ making fewer than $r$ errors, Eve will perform some optimal measurement on her system $E$ and from that gain a classical outcome $F$ which is her guess for $X$. This transforms $\rho_{XE}$ to the classical state $\tau_{XF}$ which can be represented by the probability distribution $P_{XF}$. From the data processing inequality \cite{DPI} we have
\beq
H^\epsilon_{min}(X|E)_\rho \leq H^\epsilon_{min}(X|F)_P.
\eeq
We now use the following lemma, similar to Lemma 3.1.12 from \cite{RennerThesis}:
\begin{lemma}
Let $\tau_{XF}$ be a classical state. Then the maximisation in the smooth min entropy, 
$$H^\epsilon_{min}(X|F)_\tau := \max_{\overline{\tau}_{XF}\in B^\epsilon(\tau_{XF})} \sup_{\sigma_F} H_{min}(\overline{\tau}_{XF}|\sigma_F),
$$ 
is achieved for a classical $\overline{\tau}_{XF}$ and a classical $\sigma_F$. Note that the supremum over $\sigma$ is over all density matrices with trace 1.
\end{lemma}
\begin{proof}
To prove this, we will show that for \textit{any} $\overline{\tau}^\prime_{XF} \in B^\epsilon(\tau_{XF})$ and $\sigma^\prime_F$, there exists a classical $\overline{\tau}_{XF}\in B^\epsilon(\tau_{XF})$ and a classical $\sigma_F$ such that 
$$H_{min}(\overline{\tau}_{XF}|\sigma_F) \geq H_{min}(\overline{\tau}^\prime_{XF}|\sigma^\prime_F).
$$ 
To do this, define $\mathcal{E} := \mathcal{E}_{XF}$, the projective measurement in the $XF$ basis. Choose $\overline{\tau}_{XF} := \mathcal{E}(\overline{\tau}^\prime_{XF})$ and $\sigma_F := \mathcal{E}_F(\sigma^\prime_F)$, where $\mathcal{E}_F$ is the projective measurement in the $F$ basis. Since $\mathcal{E}_F$ is a CPTP map, $\sigma_F$ still has unit trace. Also, we have 
$$ ||\overline{\tau}_{XF} - \tau_{XF}||_1 = ||\mathcal{E}(\overline{\tau}^\prime_{XF} - \tau_{XF})||_1 \leq ||\overline{\tau}^\prime_{XF} - \tau_{XF} ||_1 \leq \epsilon,
$$
 where the first equality follows from the definition of $\mathcal{E}$ and because $\tau_{XF}$ is classical. The first inequality follows because the trace distance can only decrease under CPTP maps, and the second inequality follows because $\overline{\tau}^\prime_{XF} \in B^\epsilon(\tau_{XF})$. This shows that $\overline{\tau}_{XF} \in B^\epsilon(\tau_{XF})$.

Now we use lemma 3.1.12 from \cite{RennerThesis} to say that 
$$H_{min}(\overline{\tau}_{XF}|\sigma_F) \geq H_{min}(\overline{\tau}^\prime_{XF}|\sigma^\prime_F)$$
 is true if 
 $$ 1_X \otimes \sigma_F - \mathcal{E}(1_X \otimes \sigma^\prime_F) \geq 0.
 $$ 
 Plugging in the definition of $\mathcal{E}$ we find 
 $$1_X \otimes \sigma_F - \mathcal{E}_{XF}(1_X \otimes \sigma^\prime_F) = 1_X \otimes \sigma_F - 1_X \otimes \sigma_F = 0,
 $$ 
 where we have used that $\mathcal{E}_{XF} = \mathcal{E}_X \otimes \mathcal{E}_F$ when applied to product states.
\end{proof}
This lemma means that
\beq \label{eq:probdistns}
H^\epsilon_{min}(X|F)_P = H_{min}(X|F)_{P^\prime}
\eeq
for some classical (possibly unnormalised) probability distribution $P^\prime_{XF}$. To start with, let us assume that following Eve's optimal strategy, her guess, $F$, is jointly distributed with $X$ according to $P^\prime_{XF}$. In fact, they will be distributed according to some unknown probability distribution $P_{XF}$, but $P^\prime_{XF}$ is $\epsilon$-close to $P_{XF}$ in terms of $L_1$ (or trace) distance. Note that the trace distance makes sense even for unnormalised distributions.

Let us introduce the notation
\beq
b^r_n := \sum^r_{k=0} \binom{n}{k},
\eeq
\beq
S^r_x = \{x^\prime \in \mathcal{X} : d(x, x^\prime) \leq r\},
\eeq
where $d$ is the Hamming distance. 
Under the distribution $P^\prime_{XF}$, Eve's average ``probability" (note again that $P^\prime$ may not be normalised, but we will relate it to the normalised probability distribution $P$) of making at most $r$ mistakes, $\langle p_r\rangle_{P^\prime}$, can be bounded as
\beq
\begin{split}
\langle p_r\rangle_{P^\prime} &= \sum_f P^\prime_F(f) \max_x \sum_{x^\prime\in S^r_x} P^\prime_{X|F=f}(x^\prime) \\
& \leq \sum_f P^\prime_F(f) |S^r_x| \max_{x}  P^\prime_{X|F=f}(x) \\
& = b^r_n \sum_f P^\prime_F(f) \max_{x}  P^\prime_{X|F=f}(x) \\
&:= b^r_n 2^{-H_{min}(X|F)_{P^\prime}},
\end{split}
\eeq
where $P^\prime_{F}$ is the marginal distribution of $P^\prime_{XF}$ and the last inequality follows from the definition of min-entropy on classical systems \cite{KRS2008}. 

Now, in fact the distribution shared by Bob and Eve following Eve's optimal strategy is not $P^\prime_{XF}$, but $P_{XF}$ where $P_{XF}$ is $\epsilon$-close to $P^\prime_{XF}$. We can use the above bound on $\langle p_r\rangle_{P^\prime}$ to get a bound for $\langle p_r\rangle_{P}$ as follows,
\begin{equation} \label{eq:D8}
\begin{split}
\langle p_r \rangle_{P^\prime} &= \sum_f P^\prime_F(f) \max_x \sum_{x^\prime\in S^r_x} P^\prime_{X|F=f}(x^\prime)\\
&=\sum_f P^\prime_F(f) \max_x \sum_{x^\prime\in S^r_x} \frac{P^\prime_{XF}(x^\prime, f)}{P^\prime_F(f)}\\
&=\sum_f \max_x \sum_{x^\prime\in S^r_x} P^\prime_{XF}(x^\prime,f).
\end{split}
\end{equation}
Let $f\in\{f_1, f_2,...\}$ and let $x^\prime_i$ and $x_i$ be such that
\beq
\max_x \sum_{x^\prime\in S^r_x} P^\prime_{XF}(x^\prime, f_i) = \sum_{x\in S^r_{x^\prime_i}} P^\prime_{XF}(x, f_i),
\eeq
\beq
\max_x \sum_{x^\prime\in S^r_x} P_{XF}(x^\prime, f_i) = \sum_{x\in S^r_{x_i}} P_{XF}(x, f_i),
\eeq
I.e. $x^\prime_i$ and $x_i$ specify the sets, $S^r_x$, which maximise the sum in the last equality of \eqref{eq:D8} for distributions $P^\prime_{XF}$ and $P_{XF}$ respectively. Continuing from \eqref{eq:D8} we have
\begin{equation}
\begin{split}
\langle p_r \rangle_{P^\prime} &= \sum_i\sum_{x\in S^r_{x^\prime_i}} P^\prime_{XF}(x, f_i)\\
& \geq \sum_i\sum_{x\in S^r_{x_i}} P^\prime_{XF}(x, f_i)\\
& \geq \left(\sum_i\sum_{x\in S^r_{x_i}} P_{XF}(x, f_i)\right) - \epsilon\\
&= \langle p_r \rangle_{P} -\epsilon.
\end{split}
\end{equation}
So, following her optimal strategy, we can bound Eve's average probability of making fewer than $r$ mistakes when guessing the bit values of $X$ as
\begin{equation} \label{eq:result2}
\begin{split}
\langle p_r \rangle_{P} &\leq \langle p_r \rangle_{P^\prime}+\epsilon \\
& \leq b^r_n 2^{-H_{min}(X|F)_{P^\prime}} + \epsilon \\
& \leq b^r_n 2^{-H^\epsilon_{min}(X|E)_{\rho}} + \epsilon,
\end{split}
\end{equation}
where we have used $H_{min}(X|F)_{P^\prime} = H^\epsilon_{min}(X|F)_{P} \geq H^\epsilon_{min}(X|E)_{\rho}$. 

\section{Security against repudiation}\label{App:AppendixB}
Alice aims to send a declaration $(m, Sig_m)$ which Bob will accept and which Charlie will reject. For this to happen, Bob must accept both the elements that Alice sent directly to him, and the elements that Charlie forwarded to him. In order for Charlie to reject he need only reject either the elements he received from Alice, or the elements Bob forwarded to him (or both). Intuitively, security against repudiation follows because of the symmetrisation performed by Bob and Charlie using the secret classical channel. In the distribution stage, to send the future message $m$, Alice will use the KGP with Bob and Charlie so that they hold the strings $(b_1, ..., b_L)$ and $(c_1, ..., c_L)$ respectively. We give Alice full power and assume that later on, in the messaging stage, she is able to fully control the number of mismatches her signature declaration contains with $(b_1, ..., b_L)$ and $(c_1, ..., c_L)$. Call the mismatch rates $e_B$ and $e_C$ respectively. Now, the symmetrisation process means that Bob and Charlie will randomly (and unknown to Alice) receive $L/2$ elements of the {each} other's strings. We aim to show that any choice of $e_C, e_B$ leads to an exponentially decaying probability of repudiation.\\
\\Suppose that $e_C > s_a$. In this case, Bob is selecting (without replacement) $L/2$ elements from the set $\{c_1, ..., c_L\}$, which contains exactly $e_CL$  mismatches with Alice's future declaration. The number of mismatches Bob selects then follows a hypergeometric distribution $H(L, e_CL, L/2)$ with expected value $e_CL/2$. In order to accept the message, Bob must select fewer than $s_aL/2$ errors. Using \cite{Chvatal} we can bound the probability that Bob selects fewer than $s_aL/2$ mismatches as
\begin{equation}
\begin{split}
&\mathbb{P}(\text{Bob gets fewer than $s_aL/2$ mismatches from Charlie})\\
& \leq \exp[-(e_C-s_a)^2L].
\end{split}
\end{equation}
To repudiate, Alice must make Bob accept the message, which means that Bob must accept both the part received from Alice and the part received from Charlie. Since $\mathbb{P}(A \cap B) \leq \min\{\mathbb{P}(A), \mathbb{P}(B)\}$, the probability of repudiation must be less than or equal to the above expression, and so must also decrease exponentially.\\
\\Suppose that $e_C \leq s_a$. In this case, if $e_B>s_a$, the above argument shows that it is highly likely that Bob will reject the message, so we consider only the case where $e_B \leq s_a$. Consider first the set $\{b_1, ..., b_L\}$. We can use the same arguments as above to bound the probability of selecting more than $s_vL/2$ mismatches as
\begin{equation}
\begin{split}
&\mathbb{P}(\text{Charlie gets more than $s_vL/2$ mismatches from Bob}) \\
& \leq \exp[-(s_v-e_B)^2L].
\end{split}
\end{equation}
Alice succeeds if Charlie selects more than $s_vL/2$ mismatches from either the set $\{b_1, ..., b_L\}$ or the set $\{c_1, ..., c_L\}$. Using $\mathbb{P}(A \cup B) \leq \mathbb{P}(A) + \mathbb{P}(B)$, we can see that, for the choice of $e_B, e_C \leq s_a$, we have
\begin{equation}
\begin{split}
&\mathbb{P}(\text{Charlie gets more than $s_vL/2$ mismatches}) \\
& \leq 2\exp[-(s_v-s_a)^2L].
\end{split}
\end{equation}
So again, the probability of Alice successfully repudiating decreases exponentially in the size of the signature. Similar to \cite{QDSwQKD}, Alice's best strategy would be to pick $e_B = e_C = \frac{1}{2}(s_v+s_a)$, in which case
\begin{equation}
\mathbb{P}(\text{Repudiation}) \leq 2\exp\left[-\frac{1}{4}(s_v-s_a)^2L\right].
\end{equation}

\section{Calculation of the number of quantum transmissions required per signed bit}\label{App:AppendixC}

\subsection{Parameters and constraints}
The correctness and security of the protocol depends on the three
equations \eqref{eq:honab}, \eqref{eq:forge} and \eqref{eq:rep},
which 
in turn depend on the choice of parameters $s_a$ and $s_v$. The
parameters must be such that $e^U_X < s_a < s_v < p_E$. Here, and in
all that follows, $e^U_X$ is the maximum of the worst-case error rates
Alice makes with Bob's key (found from the Alice-Bob KGP), and the
worst-case error rates Alice makes with Charlie's key (found from the
Alice-Charlie KGP). Similarly, $p_E$ is the minimum of the
eavesdropper's error rates found from the Alice-Bob and
Alice-Charlie KGP. 
The aim is to choose the parameters that minimise the number of quantum transmissions required per signed bit.
 Note that the number of quantum transmissions required per signed bit is not equal to the signature length, $L$.
  In general, due to channel losses and parameter estimation procedures, Bob will have to transmit more than $L$ quantum states to generate a signature of length $L$.

In the next section, we will calculate the length of the signature and the number of quantum transmissions
necessary to sign a message with a security level of $10^{-4}$. By
this, we mean that the probabilities of honest abort, 
forging, and repudiation, given respectively by \eqref{eq:honab},
\eqref{eq:forge} and \eqref{eq:rep}, are all less than $10^{-4}$. To find the length, {per possible one-bit message 0, 1}, of
the signature necessary to securely sign a 
one-bit message, we must first choose the parameters $s_a$ and $s_v$. 
Ideally, our choice would minimise the total length of the signature, $L$. We choose to set $\epsilon_{PE} = 10^{-5}$ and
\beq
s_a = e^U_X + \frac{p_E-e^U_X}{3}, \:\:\:\:\:\: s_v = e^U_X + \frac{2(p_E-e^U_X)}{3}.
\eeq
We note here that this may not be the optimal choices of these parameters, however, it seems natural to choose the parameters in order to equally partition the gap between $e^U_X$ and $p_E$. Nevertheless, more sophisticated optimisation of the parameters may lead to better results.

\subsection{The number of quantum transmissions required per signed bit}
In this section, we use experimental data provided by
\cite{Lucamarini} to give a rough estimate of the number of states {(per possible message bit value)}
Bob needs to transmit over a
$50$ km quantum channel to securely sign a one-bit message.
We set $\epsilon_{PE} = 10^{-5}$ in all equations that follow. The experiment in \cite{Lucamarini} approximately achieves the values
\begin{itemize}
\item Source: 1GHz pulse rate
\item Basis probabilities: $p_X = 93.75\%$, $p_Z=6.25\%$.
\item Intensity levels:\\ $(u_1, u_2, u_3) = (0.425, 0.0435, 0.0022)$.
\item Dark count rate: $p_d = 2.1\times10^{-5}$
\item Detector Efficiency: $\eta_{det}=20.4\%$
\item Channel attenuation: $0.2$dB/km
\item Receiver loss: $2.8$dB
\item Optical bit error rate: $X$ basis $Q_X=1.38\%$, $Z$ basis $Q_Z=0.76\%$
\end{itemize}

As in \cite{CaiScarani}, we model the detection rates for intensity $u_k$ as 
\beq \label{eq:det}
R_{u_k} = 1-(1-2p_d)e^{-u_k\eta}
\eeq
and the $Z$ basis bit error rates as
\beq \label{eq:err}
e_{Z,u_k} = \frac{(1-e^{-u_k\eta})Q_Z + e^{-u_k\eta}p_d}{R_{u_k}},
\eeq
and similarly for the X basis bit error rates, but using $Q_X$ in place of $Q_Z$. Over $50$~km, the attenuation due to channel and detector loss is $\eta_{ch} = 10^{-1.28} = 0.0525$. $\eta$ represents the overall transmission in the system, with $\eta = \eta_{det}\eta_{ch} = 0.0107$.

If we choose intensities with probabilities $p_{u_1}=25\%$, $p_{u_2}=40\%$ and $p_{u_3}=35\%$, then if Bob transmits $6.3\times10^{8}$ states in total, we expect the raw key to contain $8.10\times 10^{5}$ bit values from $X$ basis measurement outcomes. Of these, Bob will randomly choose $L/2= 3.86\times 10^{5}$ to be $X_{B, keep}$, another $L/2$ will be used as $X_{B, forward}$ and the remaining $k = 3.86\times 10^{4}$ will be used to estimate the correlation between Alice and Bob's X basis measurement outcomes.

For the given intensity choice probabilities and error rates, we expect to observe an $X$ basis bit error rate of $2.87\%$. We can then use Eq. \eqref{eq:pe} to upper bound the true error rate as $e^U_X = 4.02\%$.

Using Appendix \ref{App:estimates} and the detection/error rates given by \eqref{eq:det}, \eqref{eq:err} above, we can calculate the min-entropy. Setting $\epsilon=10^{-10}$ we use \eqref{eq:coh_min_entr} to find
\beq \label{eq:num_min_ent}
H^\epsilon_{min}(X|E) = 1.40\times10^5,
\eeq
Using \eqref{eq:EveError} we find $p_E = 6.96\%$, and so have $s_a = 4.99\%$ and $s_v=5.96\%$. Setting also $a=10^{-5}$ and putting these values into equations \eqref{eq:honab}, \eqref{eq:forge} and \eqref{eq:rep} we find

\beq
\mathbb{P}(\text{Honest Abort}) \leq 2\epsilon_{PE} = 2.00 \times 10^{-5},
\eeq

\beq
\mathbb{P}(\text{Forge}) \leq \epsilon_F + a + 8\epsilon_{PE} = 1.00\times 10^{-4},
\eeq

\beq
\begin{split}
\mathbb{P}(\text{Repudiation}) &\leq 2\exp\left(-\frac{1}{4}(s_v-s_a)^2L\right) \\
&= 2.97\times 10^{-8}.
\end{split}
\eeq
Thus we can see that when $6.3\times10^8$ states are transmitted, the protocol is secure to a level of $10^{-4}$. It should be stressed that this analysis is rough, and has not been optimised.

\bibliographystyle{apsrev4-1}

\begin{thebibliography}{10}
\bibitem{footnote1}  In practice, participants could use the bits they
  shared to further check the validity of the signature. This
    has some subtle security advantages even in the three-party case, and could protect against forms of collusion in the multi-party case. However, we do not specify this here, since it is not necessary under our three-party security assumptions, and it simplifies our security analysis to consider only the symmetrized keys.
\end{thebibliography}

\end{document}